\documentclass[11pt]{asaproc}
\usepackage[utf8]{inputenc}
\usepackage[T1]{fontenc}
\usepackage[english]{babel}
\usepackage{graphicx}
\usepackage{natbib}
\usepackage{amsmath,amssymb,amsthm}
\usepackage[hidelinks]{hyperref}

\usepackage{times}

\title{Searching for Gene Sets with Mutually Exclusive Mutations}

\newcommand*\samethanks[1][\value{footnote}]{\footnotemark[#1]}
\author{Paul Ginzberg\thanks{Imperial College London, 180 Queen's Gate, London SW7 2AZ, UK} \and  Federico Giorgi\thanks{Columbia University, 1130 St. Nicholas Avenue, New York, NY 10032, USA} \and  Andrea Califano\samethanks}

\newtheorem{lemma}{Lemma}

\DeclareMathOperator{\prob}{P}

\DeclareMathOperator{\coverage}{\Gamma}
\DeclareMathOperator{\overlap}{\omega}
\DeclareMathOperator{\hyperg}{Hypergeom}
\newcommand{\kmax}{k_{\mathrm{max}}}
\usepackage{color}

\newcommand{\bA}{\mathbf{A}}
\newcommand{\bC}{\mathbf{C}}

\begin{document}

\maketitle

\begin{abstract}
Cancer cells evolve through random somatic mutations. "Beneficial" mutations which disrupt key pathways (e.g. cell cycle regulation) are subject to natural selection. Multiple mutations may lead to the same "beneficial" effect, in which case there is no selective advantage to having more than one of these mutations. Hence we are interested in finding sets of genes whose mutations are approximately mutually exclusive (anti-co-occurring) within the TCGA Pancancer dataset. In principle, finding the best set is NP Hard. Nevertheless, we will show how a new Mutation anti-co-OCcurrence Algorithm (MOCA) provides an effective greedy search and testing algorithm with guaranteed control of the familywise error rate or false discovery rate, by combining some under-appreciated ideas from frequentist hypothesis testing. These ideas include: (a) A novel exact conditional test for the tendency of multiple sets to have a large/small union/intersection, which generalises Fisher's exact test of 2x2 tables. (b) Randomised hypothesis tests for discrete distributions. (c) Stouffer's method for combining p-values. (d) Weighted multiple hypothesis testing. A new approach to setting a-priori weights which generates additional implicit hypothesis tests is suggested, and allows us to preserve almost all statistical power when testing pairs despite introducing a combinatorially large number of additional hypotheses.
\begin{keywords}
hypothesis testing, Fisher's exact test, cancer, genetics, co-occurrence, exclusivity
\end{keywords}
\end{abstract}

\section{Introduction\label{intro}}

Modern sequencing technology provides a wealth of genetic, genomic and metabolomic data, and using this data to help map out the complex interactions between genes, and between genes and cancer progression is an ongoing effort. This paper describes an algorithm for this based on detecting exclusivity patterns between somatic mutations.

It is widely believed that cancer progression is linked to the gradual accumulation of somatic mutations in tumour cells, and in particular of mutations which disrupt the functioning of certain key pathways such as cell cycle regulation and DNA repair \citep{hanahan_hallmarks_2011}. These pathway-disrupting mutations allow the cells which acquire them to multiply faster, thus conferring a selective advantage relative to the surrounding cells, and hence these mutations will be more common in the population of cancer cells than what the background mutation rate on its own would predict. Passenger mutations which do not affect cancer progression will on the other hand occur only at the background rate. This underlies analyses such as  \citet{kandoth_mutational_2013} which identify genes relevant to cancer through their individual somatic mutation rates

The functioning of each pathway is complex, involving a large number of genes, some of which perform gene regulation functions. Hence each pathway may be disrupted by different somatic mutations in different patients. Although the first mutation to disrupt a given pathway offers a selective advantage, we expect that the accumulation of additional mutations in that pathway will typically not. Hence cancer cells in which multiple mutations disrupt one pathway will be rarer than if such mutations were statistically independent. In an idealised scenario, one would expect that the set of somatic mutations in each patient will consist of exactly one mutation disrupting each of the pathways. Across patients, the mutations affecting a given pathway would then be perfectly mutually exclusive. 

Reality is of course more complex than the simplified description given above. For example certain mutations may disrupt multiple pathways (e.g. TP53), and certain pairs of mutation exhibit synthetic lethality \citep{mccarthy_systems_2011,li_syn-lethality:_2014}. Because of various reasons including background mutation rates, possible residual selective advantages in having more than one mutation per pathway, sequencing errors, and the existence of additional pathway disruption mechanisms, some patients may have multiple mutations from genes in a given pathway, and some may have none. A tendency for mutations in some of the genes affecting a given pathway to co-occur less than one would expect by chance nevertheless will remain, and can be detected using statistical techniques.
Various analyses and interpretations of both co-occurrence and anti-co-occurrence patterns in somatic mutations have been suggested, mostly based on testing pairs of genes \citep{yeang_combinatorial_2008,cui_network_2010,gu_network_2013,wang_finding_2011,gu_systematic_2010}. In this paper we are interested only in detecting anti-co-occurrence, i.e. appproximate mutual exclusivity between genes, and we wish to consider not just pairs but also larger sets of genes.

Multiple methods have been suggested to search for sets of anti-co-occurring genes \citep{leiserson_comet:_2015-2,leiserson_simultaneous_2013,vandin_novo_2012,babur_systematic_2015,constantinescu_timex:_2015}. The Dendrix algorithm of \cite{vandin_novo_2012} showed that once the tendency of driver mutations to belong to mutually exclusive sets is taken into account, even relatively rare driver mutations can be identified `de novo' from somatic mutation data, i.e. without any prior pathway information. Dendrix searches for sets of genes with high overall coverage and low overlap. This approach has been further developed in the Multi-Dendrix algorithm \citep{leiserson_simultaneous_2013} which finds multiple gene sets simultaneously and the group's current state of the art CoMEt (formerly Dendrix++) \citep{leiserson_comet:_2015-2} which uses a statistical test of exclusivity to score sets. See \citet{babur_systematic_2015} for a review and comparison of methods.

In Section~\ref{sec:notation} we introduce some basic notation. Section \ref{sec:testing} describes how we compute the statistical significance of the anti-co-occurrence of a set of genes/alterations by generalising Fisher's exact test from pairs of alterations to sets of alterations, and applying a randomised Stouffer's method to combine information from multiple tumour types. Section~\ref{sec:weights} introduces a novel a weighted Bonferroni correction scheme to control the familywise error rate whilst maintaining high power for small sets. Section~\ref{sec:greedy} describes a greedy algorithm to generate a shortlist of candidate gene sets which will be tested according to Section~\ref{sec:testing}. Finally Section~\ref{sec:TCGA} shows the pattern of anti-co-occurrence detected by our method on a combined dataset of 5807 tumour samples from 22 tumour types. R code implementing our method and the full list of 654 statistically significant gene sets are available at \url{https://github.com/PaulGinzberg/MOCA/}.

\section{Setup and Notation\label{sec:notation}}

For each patient in the data we identify for each gene the presence or absence of three possible alterations: whether there are any functional SNVs (single nucleotide variations) or CNVs (copy number variations), which may be AMPs (amplifications) or DELs (deletions). This produces  an $m \times n$  binary matrix $\bA=(a_{ij})$ where each patient or sample corresponds to a column, and each gene corresponds to (up to) three rows, one for each of the three types of alterations considered (SNV, AMP, DEL). Hence $a_{ij}=1$ means ``alteration $i$ is present/mutated in sample $j$''.
For the purpose of understanding the statistical methods, one may conflate genes with alterations, i.e. assume that each row in $\bA$ corresponds to a different gene. After pre-processing, which is detailed in Section~\ref{sec:preprocessing}, our data contains $m=1418$ alterations and a total of $n=5807$ samples.
These $5807$ samples are split across 22 different types of cancer, as described in Table~\ref{table:types}.

For ease of comparison, we will follow a notation similar to \citet{leiserson_simultaneous_2013}. 
$M \subseteq \{1,\ldots,m\}$ will denote a set of alterations of size $|M|$.
 The coverage $\coverage(i)$ of alteration $i$ is the number of samples for which this alteration is present $\coverage(i)=\left|\left\{j:a_{ij}=1\right\}\right|=\sum_{j=1}^n a_{ij}$. Let $\bC=(\coverage(1),\ldots,\coverage(m))$ denote the vector of these marginal coverages. Alteration set $M$ is considered to be present/mutated in sample $j$ if at least one of the alterations from $M$ is present in sample $j$.
Hence the coverage of $M$ is defined as
\begin{equation}\label{eq:coverage}
\coverage(M)=\left|\left\{j:\sum_{g \in M} a_{gj} \geq 1\right\}\right|.
\end{equation}
 
 $M$ is perfectly mutually exclusive iff no sample has more than one alteration from $M$, i.e. iff $\coverage(M)=\sum_{g \in M} \coverage(g)$.
The overlap of an alteration set $M$ is the number of additional mutations compared to the perfect mutual exclusivity case $\overlap(M)=\sum_{g \in M} \coverage(g) - \coverage(M)$.

\section{Testing for group-wise anti-co-occurrence}\label{sec:testing}
\subsection{An exact test for whether the intersection/union of multiple sets is smaller/larger than expected by chance\label{sec:exacttest}}
Given an alteration set $M=\{g_1,\ldots,g_{|M|}\}$ we wish to test the following hypotheses:
\begin{description}
\item[$H_0:$] The alterations $a_{g_1j},\ldots,a_{g_{|M|}j}$ occur independently.
\item[$H_1:$] The alterations $a_{g_1j},\ldots,a_{g_{|M|}j}$ tend to co-occur less frequently than if they were independent, i.e. for $i,\ell \in M$ $\prob(a_{ij}=1 \cap a_{\ell j}=1) < \prob(a_{ij}=1)\prob(a_{\ell j}=1)$.
\end{description}

When $|M|=2$ the most commonly used test for this problem is a one-sided Fisher's exact test for the $2 \times 2$ contingency table of the events $a_{g_1 j}=1$ and $a_{g_2 j}=1$. The test rejects $H_0$ for large values of the test statistic $\coverage(M)$.  The exact conditional null distribution of $\coverage(M)-\coverage(g_1)$ given $(\coverage(g_1),\coverage(g_2))$ (or equivalently given $\bC$) is Hypergeometric:
\begin{equation}\label{eq:fishercdist}
\coverage(M)\mid (\coverage(g_1),\coverage(g_2)) \,\sim\, \coverage(g_1) + \hyperg(n,n-\coverage(g_1),\coverage(g_2)),
\end{equation}
where the hypergeometric distribution has probability mass function 
\begin{equation*}
f_{\hyperg(n,k,r)}(x)=\frac{{k\choose x}{n-k\choose r-x}}{{n \choose r}},
\end{equation*}
over the support $\max(0,k+r-n) \leq x \leq \min(k,r)$.

Let us now consider the case of groupwise testing, i.e. $|M|>2$. One approach which has been suggested in the literature is do define fully parametric generative models of the mutations under both the null and alternative hypothesis and then perform an asymptotic likelihood ratio test \citep{szczurek_modeling_2014,constantinescu_timex:_2015}. This approach relies on the validity of the underlying generative models. We will instead generalise Fisher's exact test by deriving the conditional null distribution of $\coverage(M)$ given $\bC$. The use of the test statistic $\Gamma(M)$ for $|M|\geq 2$ had already been suggested by \citet{ciriello_mutual_2012}. However they relied on Monte Carlo permutation testing to compute p-values, an approach which is computationally prohibitive when the p-values are very small. We will provide an algorithm for fast computation of \emph{exact} p-values for the test statistic $\Gamma(M)$.

Assume $H_0$. We will now derive the conditional null distribution of $\coverage(M)$ given \linebreak$\coverage(g_1),\ldots,\coverage(g_{|M|})$, (or equivalently given $\bC$) by iterated convolution.
For $s=1,\ldots,|M|$ let $M^{(s)}=\{g_1,\ldots,g_s\} \subseteq M$, so that $M^{(|M|)}=M$. If $1<s\leq|M|$ then 
\begin{equation}\label{eq:iterativePMF}
\coverage(M^{(s)})\mid(\coverage(M^{(s-1)}),\coverage(g_s)) \,\sim\, \coverage(M^{(s-1)})+\hyperg(n,n-\coverage(M^{(s-1)}),\coverage(g_s)).
\end{equation}
Using \eqref{eq:iterativePMF} and the fact that under $H_0$ $\coverage(M^{(s)})$ is conditionally independent of $\bC$ given $(\coverage(M^{(s-1)}),\coverage(g_s))$, we can compute the conditional probability mass function $f_{\coverage(M)|\bC}(x)$ iteratively as
\begin{equation}
f_{\coverage(M^{(s)})|\bC}(x)
=
\sum_y
f_{\hyperg(n,n-y,\coverage(g_s))}(x-y)f_{\coverage(M^{(s-1)})|\bC}(y).
\label{eq:fiter}
\end{equation}
The p-value for our test is then 
\begin{equation}\label{eq:pM}
p_M=\sum_{x \geq \coverage(M)} f_{\coverage(M^{(|M|)})|\bC}(x).
\end{equation}
Note that computing this p-value only requires evaluating the probability mass function for values of $x$ in the right tail at each iteration. In particular, if $M$ is perfectly mutually exclusive then each sum (\ref{eq:pM}), (\ref{eq:fiter}) contains only one term.

Because our test is based on $\Gamma(M)$, it has a simple and intuitive interpretation: ``Given $|M|$ subsets of $\{1,\ldots,n\}$ of sizes $\coverage(g_1),\ldots,\coverage(g_{|M|})$, is the size $\coverage(M)$ of their union larger than would be expected for independently drawn random subsets of those sizes?''. If we wished to detect co-occurrence rather than anti-co-occurrence could also perform the opposite one-sided test and reject for small values of $\coverage(M)$. Also, because $n-\coverage(M)$ is the size of the intersection of $|M|$ subsets of $\{1,\ldots,n\}$ of sizes $n-\coverage(g_1),\ldots,n-\coverage(g_{|M|})$, we can apply our test to intersections instead of unions. We expect that our generalisation of Fisher's exact test will have applications beyond biostatistics.

Because it requires $|M|-2$ ``convolutions'' (plus two sums), the computational complexity of our iterative approach is linear in $|M|$. \citet{leiserson_comet:_2015-2} perform exact groupwise testing but uses as a test statistic the number of samples for which exactly one alteration from $M$ is present $T(M)=\left|\left\{j:\sum_{g \in M} a_{gj} = 1\right\}\right|$. For $|M|=2$ both tests are equivalent and reduce to Fisher's exact test. The two tests are also equivalent when there is perfect mutual exclusivity. For $|M|>2$ and imperfect mutual exclusivity however the computational complexity of the enumeration technique employed by \citet{leiserson_comet:_2015-2} grows exponentially with $|M|$, and becomes impractical for large sample sizes except in the extreme tail of the distribution. In practice the two tests will tend to give similar answers, although the statistic $T(M)$ has the advantage of being more robust to the presence of hypermutated phenotypes.

Suppose for example that $g_1,g_2 \in M$ are respectively amplification and deletion of the same gene. Then it does not make sense to treat the anti-co-occurrence between $g_1$ and $g_2$ as evidence of any interesting biological relationship, since the two alterations are perfectly mutually exclusive by definition. Hence, whenever multiple alterations of a given gene are present in an alteration set, we will condition on the total coverage of these multiple alterations, and not just the coverage of each alteration. This is equivalent to treating $\{g_1,g_2\}$ as a single alteration of size $\coverage(\{g_1,g_2\})$ in our test.\footnote{The only difference is that when performing our multiple hypothesis testing correction, the set size will be the original the number of alterations, not the number of genes.}

\subsection{Combining p-values across cancer types}\label{sec:stouffer}
The mutation frequencies of certain genes is different in different cancer types, and indeed some mutations may be specific to only one type or subtype of cancer. In a dataset with multiple cancer types, this effect can on its own induce anti-co-occurrence between functionally unrelated alterations simply because the cancer types in which they are common are different. To avoid this effect, when applying our method we will compute p-values on each cancer type separately, and then combine these p-values with Stouffer's method.
Two issues must however be addressed. First, what weights to use in Stouffer's method, and second how to handle effects caused by the fact that the distribution of our test statistic $\coverage(M)$ (and hence of our p-values) is discrete.

Stouffer's method consists in computing the combined p-value
\[
p_\mathrm{Stouffer}=\Phi\left(\frac{\sum_{\tau} v_\tau \Phi^{-1}(p_\tau)}{\sqrt{\sum_{\tau} v_\tau^2}}\right),
\]
where $\Phi$ is the standard normal CDF, $v_\tau$ are (unnormalised) weights, and $p_\tau$ are the p-values from the independent tests being combined.

The asymptotically optimal weights to use in the case where $|M|=2$ (and the effect size is the same for all $\tau$) are well known, and are the inverse asymptotic standard deviation of the empirical log-odds ratio (which is a variance-stabilised parameter estimator) under the null. The (unnormalised) weight for each tumour type $\tau$ in this case is
\begin{align*}
v_{\{g_1,g_2\},\tau}=&
\left(
\frac{n_\tau}{\coverage_\tau(g_1)\coverage_\tau(g_2)}
+\frac{n_\tau}{\coverage_\tau(g_1)(n-\coverage_\tau(g_2))}
\right.\\
&
\left.
+\frac{n_\tau}{(n-\coverage_\tau(g_1))\coverage_\tau(g_2)}
+\frac{n_\tau}{(n-\coverage_\tau(g_1))(n-\coverage_\tau(g_2))}\right)^{-\frac{1}{2}},
\end{align*}
where $n_\tau$ is the number of samples with tumour type $\tau$ and the coverages $\coverage_\tau$ are computed using only those samples with tumour type $\tau$.

Finding a nice closed-form formula for the more general case $|M| >2$ is non-trivial since it requires us to formulate some appropriate parametric alternative distribution generalising the non-central hypergeometric distribution, and we opt instead for the following heuristic weight $v_M$, based on assuming that the power of the groupwise test can be approximated as the power of combining all pairwise tests between genes in the group, as though they were independent:
\[
v_{M,\tau} = \left( \sum_{k=1}^{|M|-1}\sum_{\ell=k+1}^{|M|} v_{\{g_k,g_\ell\},\tau}^{2}\right)^{\frac{1}{2}}.
\]
The fact that these weights are based on approximate power does not make our testing procedure approximate, although it may cause the overall power of our combined test to be lower than it would be with optimal choices of weights. Although for each $M$ our weights depend on the data, they are valid because they only use information in $\bC$, and the tests being performed are all conditional on fixed $\bC$.

In the case of exact tests based on continuous test statistics, the null distribution of each $p_\tau$ will be $\mathrm{Uniform}(0,1)$, and so will the null distribution of $p_\mathrm{Stouffer}$. This is no longer true with discrete tests. If we apply Stouffer's method naively, then we will optain a p-value satisfying $\prob(p_\mathrm{Stouffer}\leq \alpha)\leq \alpha$, but it may be catastrophically conservative. Indeed, if even just one of the tests being combined returns a p-value of 1, then the overall p-value is 1. 

\citet{kincaid_combination_1962} considers various approaches for solving this problem when Fisher's method for combining p-values is used. However, their remarks are also valid for Stouffer's method. One of the simpler suggested approaches is a simplified version of Lancaster's procedure, which is equivalent to using mid-p-values $p^\mathrm{mid}_\tau=p_\tau - 0.5(p_\tau-p^-_\tau)$, where $p^-_\tau=\prob_0(\coverage(M)>\gamma(M))$, instead of p-values $p_\tau=\prob_0(\coverage(M)\geq\gamma(M))$. The logic behind this is that the distribution of the mid-p-value is better approximated by the uniform distribution than that of the p-value. Because of its simplicity we will use this approximate approach to combine p-values from Fisher's exact test used internally in our greedy algorithm for generating candidate gene sets (See Section~\ref{sec:greedy}).
However, because this approximate approach does not guarantee $\prob(p_\mathrm{Stouffer}\leq \alpha)\leq \alpha$, (and does not produce a mid-p-value) we will use instead the Pearson approach \citep{kincaid_combination_1962} for combining the final p-values: This latter approach consists in using independently randomised p-values  
\[
p^\mathrm{rand}_\tau=p_\tau - (p_\tau-p^-_\tau) \cdot \mathrm{Uniform}(0,1)
\]
 instead of the original p-values $p_\tau$. Combining randomised p-values will never increases the number of false negatives compared to naively combining non-randomised p-values, but still controls the false positive rate so that  $\prob(p_\mathrm{Stouffer}\leq \alpha)= \alpha$. 
For large sample sizes, (and finite log-odds-ratio) this randomised $p_\mathrm{Stouffer}$ is restricted to lie in a small interval below the best (but computationally prohibitive) non-randomised p-value, so the variability introduced by randomisation becomes negligible.

\section{A novel weighted multiple hypothesis testing scheme\label{sec:weights}}

It is appropriate to correct the alteration set p-values for multiple hypothesis testing. Because the procedure for generating candidate alteration sets described in Section~\ref{sec:greedy} uses information from the data (other than $\bC$), the multiple hypothess testing correction must also take into account all tests which could have been performed, even when the number of candidate alteration sets actually tested is limited. The total number of alteration sets which could have been tested is $\sum_{k=2}^{\kmax}{m \choose k}$, where $\kmax$ denotes the maximum allowed alteration set size. If we apply a standard unweighted multiple hypothesis testing correction, e.g. a Bonferroni correction, then the corrected p-values will depend strongly on the choice of $\kmax$. If we are agnostic a-priori about the size of alteration sets and use $\kmax=m$, then the number of alteration sets which could have been tested is $2^m-m-1 > 10^{426}$. This is so large that there is little hope of any statistical significance after such a standard unweighted multiple hypothesis testing correction.

Since the main motivation behind performing groupwise testing is that the groupwise test can be much more powerful than pairwise tests, we must ensure that the loss of power caused by introducing additional hypotheses with $|M|>2$ does not overwhelm the gains.
We propose to use a \emph{weighted} multiple hypothesis testing correction scheme which upweights smaller alteration sets at the expense of larger ones, and which is based on the argument that the smallest p-values after weighting should typically correspond to the most ``interesting'' alteration sets. In practice our proposed approach leads to only a small increase in the multiple hypothesis correction applied to $p_M$ when $|M|=2$, despite the large number of additional hypotheses.

For simplicity, we will assume that there is a single cancer type in this section. Define the conditional null probability measure $\prob_0\left(\bullet\right)=\prob\left(\bullet \middle| \bC, H_\emptyset\right)$ corresponding to the global null hypothesis $H_\emptyset$ that the rows of $\bA$ are independent. Then the p-value obtained when applying our test to alteration set $M$ is 
$
p_M=\prob_0\left(\Gamma(M) \geq \gamma(M) \right)
$, where $\gamma(M)$ is the observed value of the test statistic $\Gamma(M)$.

Let $0< \alpha \leq 1$ (which need not be equal to the significance level used for selecting statistically significant alteration sets) and define the weights
\begin{equation}\label{eq:weights}
w_{M}= \frac{\left(\sum_{\ell=2}^{\kmax}{m \choose \ell}\right)\prod_{k=2}^{|M|}\left(1-(1-\alpha)^\frac{1}{m-k+1}\right)}{\sum_{\ell=2}^{\kmax}{m \choose \ell}\prod_{k=2}^{\ell}\left(1-(1-\alpha)^\frac{1}{m-k+1}\right)},
\end{equation}
which are normalised to have an average value of 1.
\begin{lemma}\label{lemma:implicittest}
\[
\prob_0\left(\exists {g \notin M} : w_{M \cup\{g\}}^{-1} p_{M \cup\{g\}} < w_M^{-1} p_M \middle| \Gamma(M)\right) \leq \alpha
\]
\end{lemma}
\begin{proof}
$\Gamma(M)$ is a random variable for each $M$, and let us first consider arbitrary fixed data leading to an observed value $\gamma(M)$. 
\begin{align*}
p_{M \cup\{g\}} &= \sum_{k \geq 0} \prob_0\left(\Gamma(M\cup\{g\}) \geq \gamma(M\cup\{g\})\middle|\Gamma(M)=k\right)\prob_0(\Gamma(M)=k)\\
&\geq \sum_{k \geq \gamma(M)} \prob_0\left(\Gamma(M\cup\{g\}) \geq \gamma(M\cup\{g\})\middle|\Gamma(M)=k\right)\prob_0(\Gamma(M)=k)\\
& \geq \prob_0\left(\Gamma(M\cup\{g\}) \geq \gamma(M\cup\{g\})\middle|\Gamma(M)=\gamma(M)\right) \sum_{k \geq \gamma(M)} \prob_0(\Gamma(M)=k)\\
&= \prob_0\left(\Gamma(M\cup\{g\}) \geq \gamma(M\cup\{g\})\middle|\Gamma(M)=\gamma(M)\right) \cdot p_M\\
&= p_{M\cup\{g\}|M} \cdot p_M,
\end{align*}
where $p_{M\cup\{g\}|M}=\prob_0\left(\Gamma(M\cup\{g\}) \geq \gamma(M\cup\{g\})\middle|\Gamma(M)=\gamma(M)\right)$ denotes the p-value obtained by a standard one-sided Fisher's exact test of independence between the events ``having a mutation for alteration $g$'' and ``having at least one mutation amongst the alterations in set $M$''.

Let us now treat the data, and hence $p_{M \cup\{g\}}$ and $p_{M\cup\{g\}|M}$, as random variables.
Note that for any $g \notin M$, $|M \cup\{g\}|=|M|+1$ and $\frac{w_{M \cup\{g\}}}{w_M}=1-(1-\alpha)^\frac{1}{m-|M|}$. Also note that that because $p_{M\cup\{g\}|M}$ is a p-value for fixed $\Gamma(M)$, $\prob_0\left(p_{M\cup\{g\}|M} \leq \alpha \middle| \Gamma(M)\right) \leq \alpha$.

\begin{align*}
&1-\prob_0\left(\exists{g \notin M} : w_{M \cup\{g\}}^{-1} p_{M \cup\{g\}} < w_M^{-1} p_M \middle| \Gamma(M)\right)\\
&= \prod_{g \notin M} \prob_0\left(w_{M \cup\{g\}}^{-1} p_{M \cup\{g\}} \geq w_M^{-1} p_M\middle| \Gamma(M)\right)\\
& \geq \prod_{g \notin M} \prob_0\left(p_{M\cup\{g\}|M} \geq \frac{w_{M \cup\{g\}}}{w_M}\middle| \Gamma(M)\right)\\
& \geq \prod_{g \notin M} \left(1- \frac{w_{M \cup\{g\}}}{w_M}\right)\\
&= 1-\alpha
\end{align*}
\end{proof}
The inequality $p_{M \cup\{g\}} \geq p_{M\cup\{g\}|M} \cdot p_M$ becomes an equality when ${M \cup\{g\}}$ has perfect mutual exclusivity; and the null distribution of the discrete p-value $p_{M\cup\{g\}|M}$ converges to a $\mathrm{Uniform}(0,1)$ for large sample sizes. Hence the inequalities in the proof of Lemma~\ref{lemma:implicittest} are ``tight'' for small $\alpha$, and any weighting scheme with weights proportional to $w_M \cdot (1+\epsilon)^{|M|}$, $\epsilon>0$ will in general no longer satisfy Lemma~\ref{lemma:implicittest}. In this sense, the weighting scheme (\ref{lemma:implicittest}) is the flattest (the one least penalising large sets) which satisfies Lemma~\ref{lemma:implicittest}.

When using a weighting scheme which does not satisfy Lemma~\ref{lemma:implicittest} for $\alpha=0.5$ (e.g. uniform weights), we run the risk that most of the highly significant aletration sets will contain passenger mutations. Hence such weighting schemes arguably favour larger alteration sets unfairly over smaller ones. In our analysis we will set $\alpha=0.05$.

We will use a weighted Bonferroni correction. Because the weights (\ref{eq:weights}) are normalised so that the average weight is $1$ this controls the familywise error rate \citep[Lemma~2.1]{roeder_genome-wide_2009}. The same weighting scheme could be used with a weighted Benjamini-Hochberg procedure to control the false discovery rate \citep[Theorem~1]{genovese_false_2006}. The weighted Bonferroni correction consists in multiplying the uncorrected p-value $p_M$ corresponding to an alteration set $M$ by
\begin{equation}\label{eq:bonferronicorrection}
\frac{\sum_{\ell=2}^{\kmax}{m \choose \ell}}{w_M}=\frac{\sum_{\ell=2}^{\kmax}{m \choose \ell}\prod_{k=2}^{\ell}\left(1-(1-\alpha)^\frac{1}{m-k+1}\right)}{\prod_{k=2}^{|M|}\left(1-(1-\alpha)^\frac{1}{m-k+1}\right)}.
\end{equation}
For small $\alpha$, large $m$, and large $\kmax$, \eqref{eq:bonferronicorrection} can be well approximated by 
\[
(\mathrm{e}^\alpha-1-\alpha)\alpha^{-|M|}|M|!{m \choose |M|}.
\]

It is clear from Table~\ref{table:bonferroni} that our weighted Bonferroni correction depends weakly on $\kmax$, and indeed when compared to the standard approach of only testing pairs of alterations ($\kmax=2$), it causes only a small increase in the multiple hypothesis correction applied to pairwise tests, even when setting $\kmax=m$. The results in Section~\ref{sec:results} show that despite almost all of the weight being focused on the smallest sets ($|M|=2$), in most cases larger sets ($3 \leq |M| \leq 6$) still obtain smaller weighted p-values than any of the pairs that they contain.
\begin{table}
\centering
\begin{tabular}{|c|cccc|}
\hline
 & $\kmax=2$ & $\kmax=3$ & $\kmax=4$ & $\kmax \geq 5$ \\\hline
$|M|=2$ & $\mathbf{1.004653 \cdot 10^6}$ & $1.021830 \cdot 10^6$ 
& $1.022050 \cdot 10^6$ 
& $1.022053 \cdot 10^6$ 
\\
$|M|=3$ &$\infty$  & $2.820910 \cdot 10^{10}$ 
& $2.821518 \cdot 10^{10}$ 
& $2.821524 \cdot 10^{10}$ 
\\
$|M|=4$ & $\infty$  & $\infty$  & $7.783707 \cdot 10^{14}$ & $7.783725 \cdot 10^{14}$
\\
$|M|=5$ & $\infty$  & $\infty$  & $\infty$  & $2.145775 \cdot 10^{19}$
\\\hline
\end{tabular}
\caption{\label{table:bonferroni}
Values for the weighted Bonferroni Correction (\ref{eq:bonferronicorrection}) where  $m=1418$ and $\alpha=0.05$. The case of unweighted pairwise tests is in bold.}
\end{table}

\section{A greedy algorithm for generating candidate sets\label{sec:greedy}}
Because even for moderate $\kmax$ the number of possible alteration sets to consider $\sum_{k=2}^{\kmax}{m \choose k}$ is too large, we cannot exhaustively compute all $p_M$, and must select a limited number of candidate alteration sets to test. We set the maximum alteration set size to $k_\mathrm{max}=10$. The number of alteration sets satisfying $2 \leq |M| \leq 10$ is $\sum_{k=2}^{10} {{1418}\choose{k}} \approx 8.84\cdot 10^{24}$.

\citet{constantinescu_timex:_2015} generate candidate sets by considering the cliques in a graph where the edges indicate statistically significant pairwise anti-co-occurrence. \citet{leiserson_comet:_2015-2} explores the space of possible sets of sets of alterations through an MCMC algorithm which swaps out one alteration at a time. As a simpler alternative to MCMC, \citet{vandin_novo_2012} also considers a greedy approach to finding the alteration set which maximises their test statistic $W(M)=\coverage(M)-\omega(M)$, and show that for large sample sizes the greedy approach succeeds with high probability.

The greedy algorithm for generating candidate alteration sets is based on two insights: Firstly, each alteration can also be considered as an alteration set of size 1. Secondly, Fisher's exact test can be used to test for anti-co-occurrence between alteration sets, i.e. \eqref{eq:fishercdist} remains valid if $M=g_1 \cup g_2$ where $g_1$ and $g_2$ are alteration sets rather than single alterations.
The proposed algorithm takes as inputs the binary alteration matrix $\bA$ and the number of alteration sets one wishes to generate maxIter (which we will set to 5000). It then generates at each iteration a new alteration set by taking the union of the two existing alteration sets which are the most significantly anti-co-occurring. Pairs of alteration sets whose union is an existing alteration set or is larger than $\kmax$ are ignored. To ensure that the space of small alteration sets is properly explored, the iterations are split into $k_\mathrm{max}-1$ equal-sized epochs and the maximum allowed size of new alteration sets is increased by 1 from 2 to $k_\mathrm{max}$ between epochs.

As described in Section~\ref{sec:stouffer}, within this greedy algorithm significance is measured by performing Fisher's exact tests independently on each of the tumour types and combining the mid p-values of these tests. An advantage of using mid-p-values over p-values here is that the ordering of mid-p-values is more appropriate for our purposes: If $g_1$ and $g_2$ are perfectly mutually exclusive, but their p-value is close to 1 because of discreteness, then they are arguably more significantly anti-co-occurring than two alterations with a large overlap. The mid-p-value of perfectly mutually exclusive alterations will always be $\leq 0.5$.
The fact that the tests used in the greedy algorithm may be liberal is not an issue since the purpose of this step is only to heuristically search for candidate alteration sets on which accurate groupwise testing will be performed.

All subsets of the candidate alteration sets are then added to the set of candidate alteration sets.

Finally, the exact p-value $p_\tau$ for each candidate alteration set is computed for each tumour type as described in Section~\ref{sec:exacttest}, and these are randomised and combined as described in Section~\ref{sec:stouffer} to produce an overall p-value $p_M=p_M^\mathrm{rand}$. The Bonferroni correction \eqref{eq:bonferronicorrection} with $\alpha=0.05$ is then applied, and Bonferroni-corrected p-values greater than $\alpha=0.05$ are discarded as non-significant. The ordering of Bonferroni-corrected p-values is the  same as the ordering of weighted p-values, and in light of Lemma~\ref{lemma:implicittest} we also discard an alteration set as non-significant if any of its subsets (or supersets) has a smaller Bonferroni-corrected p-value.

\section{Application to the TCGA Pancancer dataset\label{sec:TCGA}}
\subsection{Pre-processing}\label{sec:preprocessing}
Somatic SNP and CNV data is obtained for 22 tumour types as described in Table~\ref{table:types}. All but one of the tumour type datasets (and more than 99\% of samples) was collected by TCGA. The total sample size across all 22 datasets is $n=5807$.

\begin{table}[h]
\centering
\begin{tabular}{llrl}
Tumour Type & Acronym & Samples & Publication\\\hline

Acute myeloid leukemia & AML & 193 & \citep{tcga_genomic_2013}\\ 

Urothelial bladder cancer & BLCA & 129 & \citep{tcga_comprehensive_2014}\\
Breast cancer & BRCA & 975 & \citep{tcga_comprehensive_2012}\\
Colon adenocarcinoma & COAD & 216 & \citep{tcga_comprehensive_2012-1}\\ 
Glioblastoma multiforme & GBM & 281 & \citep{tcga_comprehensive_2008}\\
Head \& neck squamous cell carcinoma & HNSC & 302 & \citep{tcga_comprehensive_2015}\\
Clear cell kidney carcinoma & KIRC & 292 & \citep{tcga_comprehensive_2013}\\
Papillary kidney carcinoma & KIRP & 161 & \citep{tcga_comprehensive_2016}\\
Lower grade glioma & LGG & 286 & \citep{tcga_comprehensive_2015-1}\\
Liver hepatocellular carcinoma & LIHC & 187 & \\
Lung adenocarcinoma & LUAD & 466 & \citep{tcga_comprehensive_2014-1}\\
Lung squamous cell carcinoma & LUSC & 178 & \citep{tcga_comprehensive_2012-2}\\
Ovarian serous cystadenocarcinoma & OV & 139 & \citep{tcga_integrated_2011}\\
Pancreatic ductal adenocarcinoma & PAAD & 168 & \\

Prostate cancer (castration-resistant) & CRPC & 58 & \citep{grasso_mutational_2012}\\ 

Prostate adenocarcinoma & PRAD & 199 & \citep{TCGA_molecular_2015}\\
Rectal adenocarcinoma & READ & 81 & \citep{tcga_comprehensive_2012-1}\\ 
Sarcoma & SARC & 257 & \\
Cutaneous melanoma & SKCM & 364 & \citep{tcga_genomic_2015}\\
Stomach adenocarcinoma & STAD & 229 & \citep{tcga_comprehensive_2014-2}\\
Papillary thyroid carcinoma & THCA & 403 & \citep{tcga_integrated_2014}\\
Uterine corpus endometrial carcinoma & UCEC & 243 & \citep{tcga_integrated_2013}\\\hline
\end{tabular}
\caption{Number of samples for each tumour type. All data is from The Cancer Genome Atlas Research Network, with the exception of the additional prostate cancer dataset from Grasso et al.\label{table:types}}
\end{table}

\paragraph{Selection of genes:}
Although the proposed algorithm can be used in theory to detect genes which are relevant to cancer purely through their mutation pattern relative to other genes, this will not be attempted in this paper, and we will focus instead only on the search for relationships between genes. Hence we will restrict our analysis to genes which have already been identified as potentially relevant to cancer because of their high mutation rate, or through other means. Specifically, we include in our analysis the 127 genes which were identified as significantly mutated in \cite[Supplementary Table 4]{kandoth_mutational_2013} and the 487 genes
 listed by the Catalogue Of Somatic Mutations in Cancer (COSMIC) cancer gene census \cite[Supplementary Table S1]{futreal_census_2004}.
 Combining these two sources yields a whitelist of 547 genes. 
 Because there are three possible alteration types (SNV, AMP, DEL) for each gene, the whitelist corresponds to 1641 possible alterations.

\medskip

 Alterations which are mutated in exactly the same set of samples (i.e. identical rows in $\bA$) are merged into one since in terms of their anti-co-occurrence they are mathematically indistinguishable. This step can reduce the computational requirements and can also aid interpretation. None of these merged alterations appear in the significant sets.
 
 The number of alterations considered is reduced further by eliminating those alterations which are too rare for there to be a realistic chance of measuring statistically significant anti-co-occurrence after a Bonferroni correction is applied. The minimum number of mutations required for inclusion is set to be greater than $\log_2(\mbox{Number of alterations remaining}-1)$. Based on this heuristic, all alterations with $\coverage(g)<11$ were removed. This leaves us with a final binary mutation matrix $\bA$ describing $m=1418$ alterations across $n=5807$ samples.
 
 
Missing values in $\bA$ are set to 0 (i.e. the alteration is assumed to be absent).

\subsection{Results}\label{sec:results}

After a weighted Bonferroni correction with $\alpha=0.05$ as the significance level, We find 654 statistically significant sets of anti-co-occurring alterations (6 pairs, 147 triplets, 222 quadruplets, 261 quintuplets and 18 sextuplets). The 10 most statistically significant sets are given in Table~\ref{table:top10}. These alteration sets contain a total of 125 different alterations (116 different genes). We can think of each alteration set as a fully connected graph with alterations as vertices. The union of these 654 graphs contains 125 vertices and 969 edges. This graph is displayed in Figure~\ref{fig:fullgraph}. The rarest alteration out of these 125 is TNFAIP3(D) with 35 mutations (0.6\%).
\begin{figure}[htbp]
\includegraphics[width=0.9\textwidth,trim=90 90 70 80, clip]{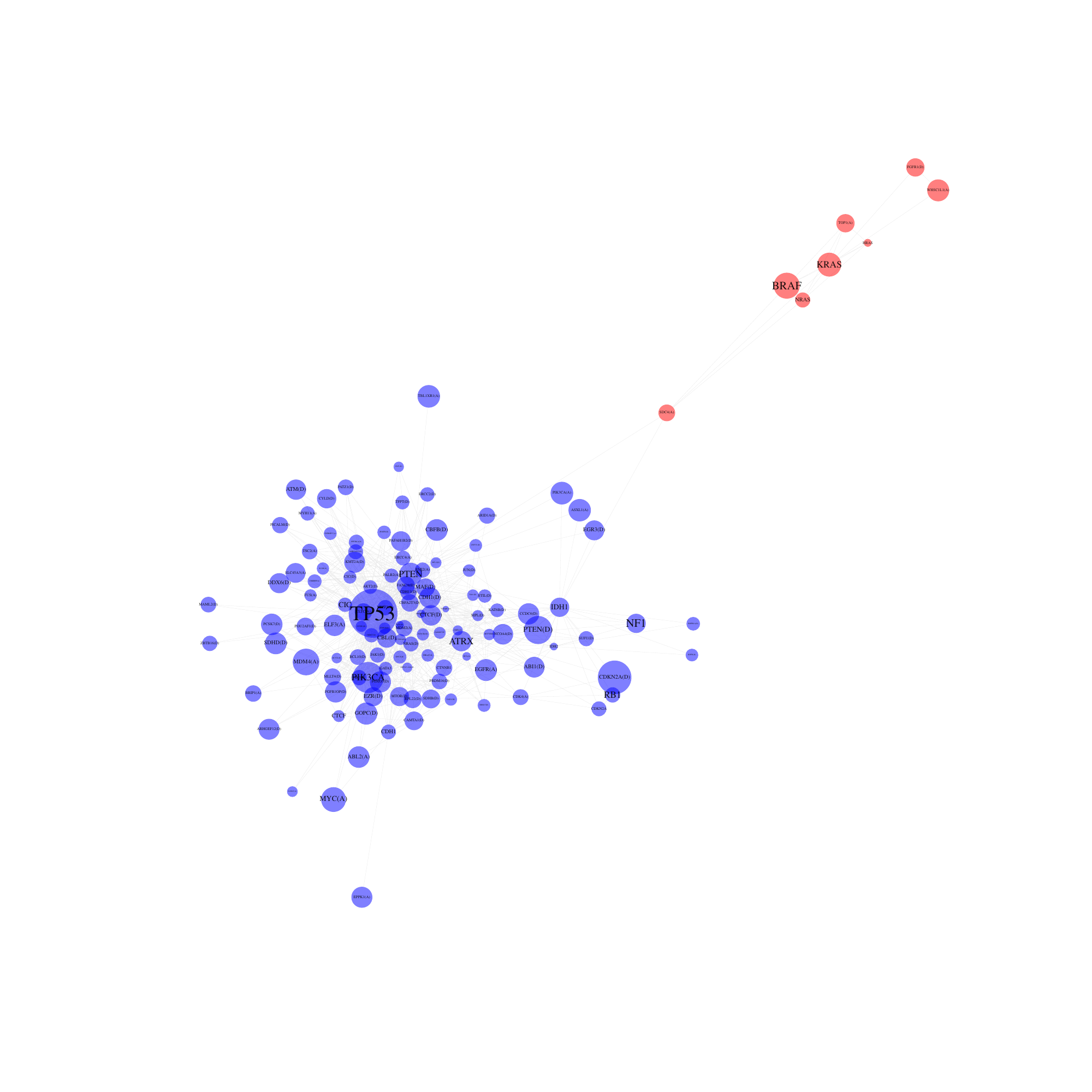}
\caption{The graph obtained from the union of the 654 statistically significant ($\alpha=0.05$) alteration sets obtained with our approach. The size of nodes is proportional to the coverage of the alteration. The color of nodes indicates the two communities (detected by label propagation).\label{fig:fullgraph}}
\end{figure}

The increased power from using groupwise testing is evident if we compare with the results from a pairwise analysis.
If we compute pairwise p-values based on Fisher's exact test (combined over all cancer types)%
,\footnote{In the pairwise analysis we used mid-p-values instead of randomised p-values in Stouffer's method.}
then after applying the Bonferroni correction $\frac{m(m-1)}{2}$, only 82 pairs of alterations are significant at the level $\alpha=0.05$, between 67 unique alterations (63 genes). The rarest alteration out of these 67 is TLX1(D) with 59 mutations (1\%). Only one of the pairs detected by the pairwise approach is missing an edge in the graph obtained with the groupwise approach, namely  (EIF4A2(A),PTEN), and this is because EIF4A2(A) is not in any of the sets detected by the groupwise approach. If we take the graph union of these pairs, displayed in Figure~\ref{fig:pairwisegraph}, there are no cliques of size $\geq 3$, i.e. the pairwise approach fails to detect any groups of size $\geq 3$. The graph does however split into 7 connected components.
\begin{figure}[htbp]
\includegraphics[width=0.9\textwidth,trim=100 100 80 90, clip]{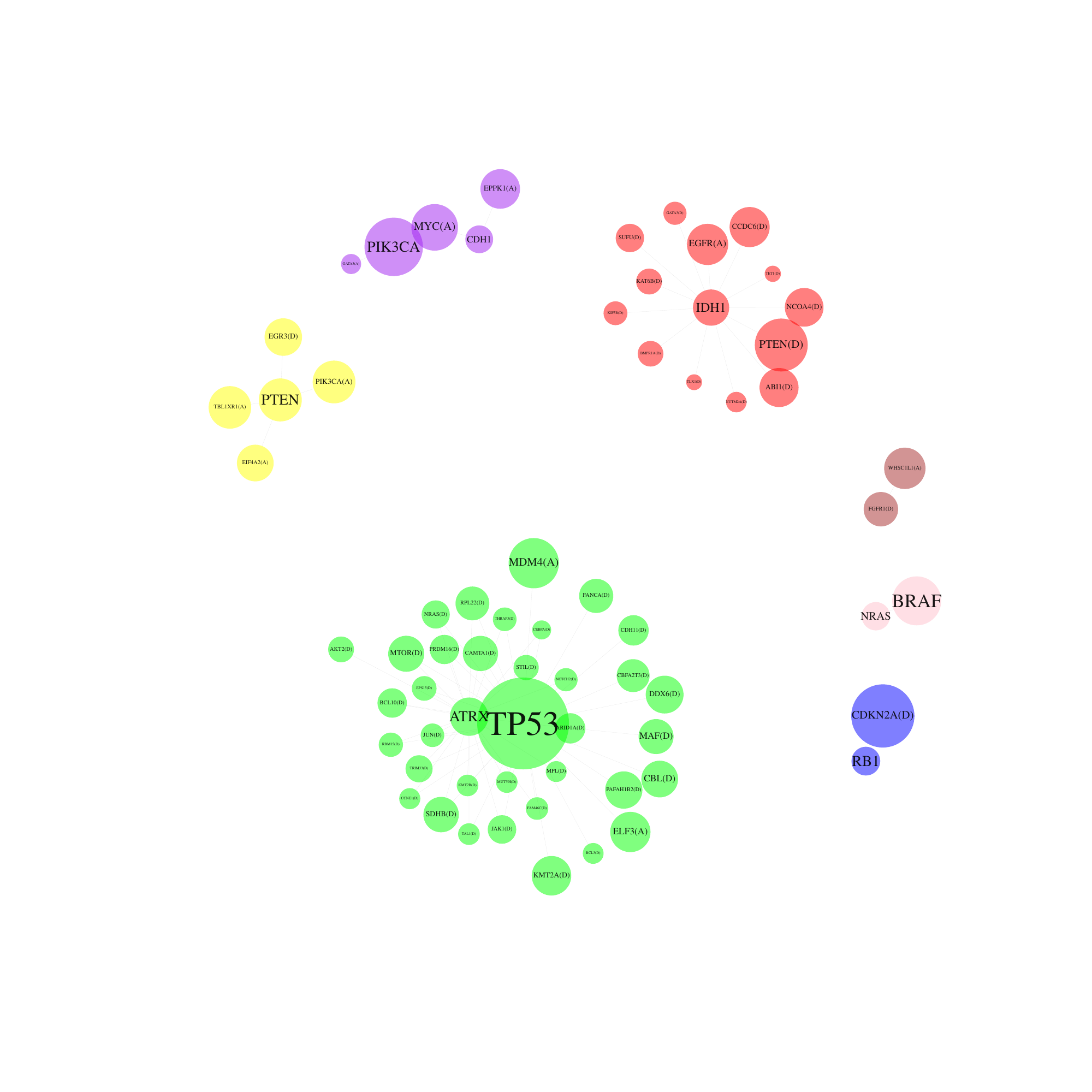}
\caption{The graph obtained from the statistically significant pairs ($\alpha=0.05$) of anti-co-occurring alterations. The size of nodes is proportional to the coverage of the alteration. The color of nodes is based on the connected components.\label{fig:pairwisegraph}}
\end{figure}

\begin{table}[htbp]
\centering
\begin{tabular}{rrrrl}
  \hline
 &$\frac{1}{n}\coverage(M)$& \begin{tabular}{@{}c@{}}p-value\\(uncorrected)\end{tabular} & \begin{tabular}{@{}c@{}}p-value\\(Weighted\\ Bonferroni)\end{tabular} & Alteration Set \\ 
  \hline
  1 & 25.6\% & $1.22 \cdot 10^{-53}$ & $5.20 \cdot 10^{-39}$ & TOP1(A), BRAF, HRAS, KRAS, NRAS \\ 
  2 & 24.0\% & $8.11 \cdot 10^{-45}$ & $1.22 \cdot 10^{-32}$ & SDC4(A), BRAF, KRAS, NRAS \\ 
  3 & 42.8\% & $3.77 \cdot 10^{-27}$ & $1.61 \cdot 10^{-17}$ & ARID1A(D), PTEN, TP53 \\ 
  4 & 37.2\% & $2.04 \cdot 10^{-26}$ & $8.70 \cdot 10^{-17}$ & STIL(D), HMGA2(A), TP53 \\ 
  5 & 9.2\% & $3.84 \cdot 10^{-26}$ & $1.64 \cdot 10^{-16}$ & TLX1(D), RBM15(D), ATRX \\ 
  6 & 36.3\% & $4.78 \cdot 10^{-26}$ & $2.04 \cdot 10^{-16}$ & MUTYH(D), TNFAIP3(D), TP53 \\ 
  7 & 48.1\% & $5.73 \cdot 10^{-31}$ & $2.45 \cdot 10^{-16}$ & STIL(D), TP53, MDM2(A), PTEN, CTCF(D) \\ 
  8 & 36.8\% & $6.84 \cdot 10^{-26}$ & $2.92 \cdot 10^{-16}$ & RBM15(D), TNFAIP3(D), TP53 \\ 
  9 & 13.9\% & $8.82 \cdot 10^{-26}$ & $3.76 \cdot 10^{-16}$ & FAM46C(D), NCOA4(D), ATRX \\ 
  10 & 45.3\% & $4.26 \cdot 10^{-28}$ & $6.43 \cdot 10^{-16}$ & JUN(D), FANCA(D), TP53, PTEN \\ 
   \hline
\end{tabular}
\caption{The top 10 most significant alteration sets obtained with our method after Bonferroni correction. (A) indicates an amplification, (D) indicates a deletion. The full list is available at \url{https://github.com/PaulGinzberg/MOCA/blob/master/example-output/pancancer-test2-significantv2.txt}\label{table:top10}}
\end{table}

\section{Acknowledgements}
This research was supported in part by EPSRC Mathematics Platform Grant EP/I019111/1. The empirical results in Section~\ref{sec:TCGA} are based upon data generated by the TCGA Research Network: \url{http://cancergenome.nih.gov/}.

\bibliographystyle{asa}
\bibliography{biblio}

\end{document}